\documentclass[11pt]{amsart}

\usepackage{amsmath,amssymb,amsthm,mathtools}
\usepackage[numbers,sort&compress]{natbib}
\usepackage[colorlinks=true,linkcolor=blue,citecolor=blue,urlcolor=blue]{hyperref}
\usepackage[nameinlink,noabbrev]{cleveref}

\theoremstyle{plain}
\newtheorem{theorem}{Theorem}
\newtheorem{lemma}{Lemma}

\newtheorem{corollary}{Corollary}
\newtheorem{claim}{Claim}

\theoremstyle{definition}

\newtheorem{assumption}{Assumption}

\theoremstyle{remark}

\numberwithin{equation}{section}

\title[Two Moments for Risk-Monotone Additive Statistics]{Two Moments for Risk-Monotone Additive Statistics}
\author{Mark Whitmeyer}
\date{\today}
\subjclass[2020]{Primary 60A10; Secondary 60E07}
\thanks{\emph{Acknowledgements.} Dedicated to KS. I thank Joseph Whitmeyer for his comments. I used ChatGPT as one would an RA (checking proofs, identifying references, scanning for typos, etc.) and solicited feedback from \href{refine.ink}{refine.ink}.}

\begin{document}

\begin{abstract}
Every statistic on laws with finite \(p\)th moment that is additive across independent risks and monotone in mean-preserving spreads depends only on an additive function of the mean when \(p<2\), and only on such a function and a nonnegative multiple of the variance when \(p\geq2\).
\end{abstract}

\maketitle

Variance is additive across independent risks and monotone in mean-preserving spreads. We show that these properties leave room for variance exactly when the domain requires a second moment.

Fix \(p\in[1,+\infty)\). On centered laws with finite \(p\)th moment, we show that every real-valued statistic that is additive under convolution and monotone in convex order is zero when \(p<2\), and is a nonnegative multiple of variance when \(p\geq2\). In terms of certainty equivalents, expected value is the only possibility when \(p<2\), while every possibility when \(p\geq2\) is mean-variance.

\citet{pomattoStrackTamuz2020} characterize mean-variance certainty equivalents on laws with all finite moments, and \citet{fritzMuTamuz2026} ask whether convex-order monotonicity singles out variance already on the finite-second-moment domain. We answer that question and identify \(p=2\) as the exact threshold. More broadly, the paper belongs to the study of monotone additive statistics pursued by, e.g., \citet{ruzsaSzekely1988}, \citet{mattner2004}, and \citet{muPomattoStrackTamuz2024}. The proof combines the independent-background-risk comparison of \citet{pomattoStrackTamuz2020} with the observation that any failure of \(L^p\)-continuity at zero can be amplified to an arbitrary degree by summing independent small risks.

\section{Setup and main results}

Fix \(p\in[1,+\infty)\). Let \(\mathcal P^p\) denote the set of Borel probability measures \(\mu\) on \(\mathbf R\) such that \(\int_{\mathbf R}|x|^p\,\mu(dx)<+\infty\). Write \(m(\mu)\coloneqq\int_{\mathbf R}x\,\mu(dx)\) for the mean of \(\mu\), and let \(\mathcal P_0^p\coloneqq\left\{\mu\in\mathcal P^p\colon m(\mu)=0\right\}\). For \(q\in[1,+\infty)\), let \(L^q\) denote the random variables \(X\) such that \(\operatorname{E}[|X|^q]<+\infty\), and write \(\lVert X\rVert_q\coloneqq\left(\operatorname{E}[|X|^q]\right)^{1/q}\). For \(\mu\in\mathcal P^2\), write \(\operatorname{Var}(\mu)\coloneqq\int_{\mathbf R}(x-m(\mu))^2\,\mu(dx)\). For \(X\in L^2\), write \(\operatorname{Var}(X)\coloneqq\operatorname{E}\left[\left(X-\operatorname{E}[X]\right)^2\right]\).

Write \(\mu*\nu\) for the convolution of \(\mu\) and \(\nu\), and let \(\delta_x\) denote the point mass at \(x\). If \(X\) and \(Y\) are independent with laws \(\mu\) and \(\nu\), then \(X+Y\) has law \(\mu*\nu\). Throughout, we freely pass to product extensions of probability spaces, preserving the joint laws of previously introduced random variables, whenever we need to realize auxiliary random variables or independent copies.

For \(\mu,\nu\in\mathcal P^1\), write \(\mu\preceq_{\mathrm{cx}}\nu\) if
\[
\int_{\mathbf R}\varphi(x)\,\mu(dx)\leq\int_{\mathbf R}\varphi(x)\,\nu(dx),
\]
for every convex function \(\varphi\colon\mathbf R\to\mathbf R\) with at most linear growth.\footnote{This means that there is a constant \(K<+\infty\) such that \(|\varphi(x)|\leq K(1+|x|)\) for every \(x\in\mathbf R\). This restriction on growth ensures that both integrals are finite for every \(\mu,\nu\in\mathcal P^1\).} This order implies \(m(\mu)=m(\nu)\). When \(\mu,\nu\in\mathcal P^2\), it also implies \(\operatorname{Var}(\mu)\leq\operatorname{Var}(\nu)\). If \(X\) and \(Y\) are random variables, write \(X\preceq_{\mathrm{cx}}Y\) when \(\mathcal L(X)\preceq_{\mathrm{cx}}\mathcal L(Y)\), where \(\mathcal L(X)\) denotes the law of \(X\). A statistic on a class \(\mathcal A\) of probability laws is a map \(F\colon\mathcal A\to\mathbf R\). If \(X\) is a random variable with \(\mathcal L(X)\in\mathcal A\), write \(F(X)\coloneqq F(\mathcal L(X))\).

Call a real-valued map \(F\) \textit{additive under convolution} if \(F(\mu*\nu)=F(\mu)+F(\nu)\) whenever \(\mu\), \(\nu\), and \(\mu*\nu\) belong to its domain. Call \(F\) \textit{monotone in convex order} if \(F(\mu)\leq F(\nu)\) whenever \(\mu\preceq_{\mathrm{cx}}\nu\) and both laws belong to its domain.
\begin{assumption}\label{ass:additive-convex-order}
Fix \(p\in[1,+\infty)\) and a real-valued map \(R\colon\mathcal P_0^p\to\mathbf R\). Assume that \(R\) is additive under convolution and monotone in convex order.
\end{assumption}
The assumption says that \(R(X+Y)=R(X)+R(Y)\) whenever \(X\) and \(Y\) are independent centered random variables in \(L^p\), and that \(R(X)\leq R(Y)\) whenever \(X\preceq_{\mathrm{cx}}Y\).

\begin{theorem}\label{thm:second-moment-dichotomy}
Maintain \Cref{ass:additive-convex-order}. If \(1\leq p<2\), then \(R(\mu)=0\) for every \(\mu\in\mathcal P_0^p\). If \(2\leq p<+\infty\), then there is a constant \(c\geq0\) such that \(R(\mu)=c\operatorname{Var}(\mu)\) for every \(\mu\in\mathcal P_0^p\). Conversely, the zero map\footnote{\textit{Viz.}, the map \(\mu\mapsto0\).} satisfies \Cref{ass:additive-convex-order} for every \(p\in[1,+\infty)\), and, when \(p\geq2\), every map \(\mu\mapsto c\operatorname{Var}(\mu)\) with \(c\geq0\) satisfies it.\end{theorem}
\Cref{thm:second-moment-dichotomy} concerns only laws with mean zero. For an arbitrary random variable \(X\), write \(X=\left(X-\operatorname{E}[X]\right)+\operatorname{E}[X]\). The first term is centered (mean zero), and the second is constant. Through this lens, we see that a statistic that is additive under convolution is determined by its values on centered (mean  laws and its values on point masses. The values on point masses may be any additive function of their locations, because the convex order compares only laws with the same mean.

\begin{corollary}\label{cor:full-classification}
Fix \(p\in[1,+\infty)\). Suppose that \(\Phi\colon\mathcal P^p\to\mathbf R\) is real-valued and additive under convolution, and that \(\Phi(\mu)\leq\Phi(\nu)\) whenever \(\mu\preceq_{\mathrm{cx}}\nu\). There is an additive function\footnote{That is, a function \(a\colon\mathbf R\to\mathbf R\) satisfying \(a(x+y)=a(x)+a(y)\) for every \(x,y\in\mathbf R\).} \(a\colon\mathbf R\to\mathbf R\) such that
\[
\Phi(\mu)=
\begin{cases}
a(m(\mu)), \quad &\text{if} \quad 1\leq p<2,\\
a(m(\mu))+c\operatorname{Var}(\mu),\quad &\text{if} \quad 2\leq p<+\infty,
\end{cases}
\]
where \(c\geq0\) in the second case. Conversely, every map of the displayed form is additive under convolution and monotone in convex order.
\end{corollary}

In the next corollary we impose \(C(\delta_x)=x\), which makes it so that every certain amount \(x\) is assigned the value \(x\). Apply \Cref{cor:full-classification} to \(-C\). Since \(\delta_x\) has mean \(x\) and variance zero, the normalization requires the additive function in that classification to satisfy \(a(x)=-x\). Multiplying the resulting formula by \(-1\) leaves \(m(\mu)\) as the mean term in \(C\) and places a minus sign before the variance term.

\begin{corollary}\label{cor:certainty-equivalent}
Fix \(p\in[1,+\infty)\). Suppose that \(C\colon\mathcal P^p\to\mathbf R\) is real-valued, \(C(\delta_x)=x\) for every \(x\in\mathbf R\), and \(C(\mu*\nu)=C(\mu)+C(\nu)\) for all \(\mu,\nu\in\mathcal P^p\). Suppose also that \(C(\mu)\geq C(\nu)\) whenever \(\mu\preceq_{\mathrm{cx}}\nu\). If \(1\leq p<2\), then \(C(\mu)=m(\mu)\) for every \(\mu\in\mathcal P^p\). If \(2\leq p<+\infty\), then there is a constant \(c\geq0\) such that \(C(\mu)=m(\mu)-c\operatorname{Var}(\mu)\) for every \(\mu\in\mathcal P^p\). Conversely, expected value satisfies the hypotheses for every \(p\in[1,+\infty)\), and, when \(p\geq2\), every map \(\mu\mapsto m(\mu)-c\operatorname{Var}(\mu)\) with \(c\geq0\) satisfies them.
\end{corollary}

\section{Proof Intuition}

Quite simply, the distinction made by \Cref{thm:second-moment-dichotomy} between \(p<2\) and \(p\geq2\) arises because every risk in the domain has finite variance only in the latter case. In the proof, we first pin down the statistic \(R\) on bounded centered risks before then showing that unbounded risks engender no additional possibilities. We say in this sketch that one risk is less risky than another when it is smaller in convex order.

\emph{Monotonicity in variance via background risks.} Variance alone does not generally make two risks comparable in convex order. Nevertheless, additivity under convolution plus monotonicity allow us to connect the two. The first step in doing so is to take \citet[Theorem 2]{pomattoStrackTamuz2020} off-the-shelf to produce our immediate corollary  \Cref{lem:background-risk-comparison}, which says that if one bounded centered risk has smaller variance than another, then we can always find a common independent background that makes the first less risky. We then use this fact, then cancel out the common background risk by additivity, to show in \Cref{lem:bounded-calibration} that, among bounded centered risks, a smaller variance cannot be assigned a larger value by \(R\).

\emph{Variance encompasses all for bounded risks.} Furthermore, in \Cref{lem:bounded-calibration}, we go on to argue that two bounded centered risks with the same variance must actually have the same \(R\)-value. The proof is simple: if not, taking many independent copies would multiply the difference between their \(R\)-values (again via additivity). Adding one fixed risk to one of the sums would make its variance strictly larger, but the resulting comparison would eventually contradict the multiplied difference in values. Accordingly, \(R\) can depend only on variance (on the bounded domain). Moreover, independent addition makes this dependence additive, and we have already deduced that \(R\)'s dependence on variance must be nondecreasing. These two properties force linearity: there is a constant \(c\geq0\) such that \(R(X)=c\operatorname{Var}(X)\) for every bounded centered \(X\).

\emph{A bound for unbounded risks.} Suppose that \(p\geq2\), so every risk under consideration has finite variance, and define \(D(X)\coloneqq R(X)-c\operatorname{Var}(X)\), which measures how far an unbounded risk lies above or below the formula we just established for bounded risks. In \Cref{lem:residual-nonnegative}, we discover that it cannot lie below that formula. In the proof of this lemma, we replace the most extreme outcomes of an unbounded risk by their conditional mean. This produces bounded risks that are less risky than the original risk and whose variances approach its variance. Since the values of the bounded risks are already known, we conclude by monotonicity that \(D(X)\geq0\).

\emph{The opposite bound for unbounded risks I.}
It remains to show that \(D(X)\) cannot be strictly positive. The first of two steps in this endeavor is \Cref{lem:tail-compression}, which shows that any such extra value can be carried, up to an arbitrarily small loss, by a risk with arbitrarily small \(L^p\)-norm. We prove this via a construction that replaces the common, moderate outcomes of a risk by their average while retaining its rare, extreme outcomes. As we raise the threshold separating moderate outcomes from extreme ones, the retained extreme outcomes become increasingly rare, while the average assigned to the remaining outcomes approaches zero. As a result, the risk becomes arbitrarily small. To compare the original risk with the compressed one, we then add back an independent bounded risk that reproduces the variation removed from the moderate outcomes. \Cref{lem:bounded-calibration} tells us exactly how much this added risk contributes to \(R\). When \(p\geq2\), the compression also reduces variance by almost the same amount. Thus, after subtracting \(c\operatorname{Var}\), the two effects nearly cancel, and the compressed risk retains almost all of \(D(X)\). When \(c=0\), the added bounded risk contributes nothing to \(R\), so the compressed risk does not reduce \(R(X)\).

\emph{The opposite bound for unbounded risks II.}
Next, in \Cref{lem:stacking}, we show that the value preserved by \Cref{lem:tail-compression} must be zero. If a strictly positive value could be carried by arbitrarily small risks, we could choose infinitely many independent such risks, each producing nearly the same value, but small enough that their sum still belongs to \(L^p\). Additivity would make the value of the first \(N\) risks grow without bound. This value is always no greater than the \(R\)-value of that finite sum. Moreover, each finite sum is also less risky than the full sum, because completing the sum only adds further independent centered risks. Monotonicity would, therefore, force the full sum to have an arbitrarily large \(R\)-value, contradicting the fact that \(R\) is real-valued.

\emph{Eliminating \(c\) for low \(p\).}
It remains to determine whether \(c\) can be positive when \(p<2\). In \Cref{lem:subquadratic-calibration}, we find that it cannot. In this range,
the domain contains a centered risk with finite \(p\)th moment but infinite variance. Its bounded truncations are less risky than the original risk, while their variances grow without bound. If \(c\) were strictly positive, then \Cref{lem:bounded-calibration} would assign these truncations arbitrarily large values. Monotonicity would then require the original risk to have an arbitrarily large \(R\)-value, contradicting the fact that \(R\) is real-valued. That leaves \(c=0\) as the only possibility.

The pieces now fit together. When \(p\geq2\), \Cref{lem:residual-nonnegative} makes \(D\) nonnegative; \Cref{lem:tail-compression} preserves \(D\), up to an arbitrarily small loss, on arbitrarily small risks; and \Cref{lem:stacking} forces \(D=0\). Therefore, \(R=c\operatorname{Var}\). When \(p<2\), \Cref{lem:subquadratic-calibration} first implies \(c=0\). Convex-order monotonicity then makes \(R\) nonnegative, while \Cref{lem:tail-compression} preserves \(R\) on arbitrarily small risks. Applying \Cref{lem:stacking} directly to \(R\) yields \(R=0\).

\section{Proof of the dichotomy}

A random variable \(X\) is centered if \(\operatorname{E}[X]=0\), and bounded if there is a constant \(M<+\infty\) such that \(|X|\leq M\) almost surely. We use the following consequence of \citet[Theorem~2]{pomattoStrackTamuz2020}.

\begin{lemma}\label{lem:background-risk-comparison}
Suppose that \(X\) and \(Y\) are bounded centered random variables and \(\operatorname{Var}(X)<\operatorname{Var}(Y)\). There is a centered random variable \(Z\), independent of \((X,Y)\) and having finite moments of every order, such that \(X+Z\preceq_{\mathrm{cx}}Y+Z\).
\end{lemma}

\begin{proof}[Proof of \Cref{lem:background-risk-comparison}]
Because \(X\) and \(Y\) are centered and bounded and satisfy
\(\operatorname{Var}(X)<\operatorname{Var}(Y)\), \citet[Theorem~2]{pomattoStrackTamuz2020} provides a random variable \(\widetilde Z\), independent of the random vector \((X,Y)\) and having finite moments of every order, such that \(X+\widetilde Z\) second-order stochastically dominates \(Y+\widetilde Z\). The two sums have the same mean, so this comparison implies \(X+\widetilde Z\preceq_{\mathrm{cx}}Y+\widetilde Z\). Set \(Z\coloneqq\widetilde Z-\operatorname{E}[\widetilde Z]\). Translating both sides of a convex-order comparison by the same constant preserves the comparison, so \(X+Z\preceq_{\mathrm{cx}}Y+Z\). The variable \(Z\) is centered, remains independent of \((X,Y)\), and retains finite moments of every order because subtracting a finite constant preserves every finite absolute moment.
\end{proof}

\begin{lemma}\label{lem:bounded-calibration}
Maintain \Cref{ass:additive-convex-order}. There is a constant \(c\geq0\) such that \(R(X)=c\operatorname{Var}(X)\) for every bounded centered random variable \(X\).
\end{lemma}

\begin{proof}[Proof of \Cref{lem:bounded-calibration}]
Suppose first that \(X\) and \(Y\) are bounded and centered, with \(\operatorname{Var}(X)<\operatorname{Var}(Y)\). By \Cref{lem:background-risk-comparison}, there is a centered \(Z\in L^p\), independent of \((X,Y)\), such that \(X+Z\preceq_{\mathrm{cx}}Y+Z\). Convex-order monotonicity and convolution additivity of \(R\) yield
\[R(X)+R(Z)=R(X+Z)\leq R(Y+Z)=R(Y)+R(Z),\]
and hence
\[
\operatorname{Var}(X)<\operatorname{Var}(Y)
\quad \Longrightarrow \quad
R(X)\leq R(Y).
\tag{1}\label{eq:bounded-variance-order}
\]

Now suppose that \(X\) and \(Y\) are bounded and centered, with \(\operatorname{Var}(X)=\operatorname{Var}(Y)\). Suppose for the sake of contradiction that \(R(X)>R(Y)\). Let \(W\) be any bounded centered random variable with positive variance. For every \(n\), take \(X_1,\ldots,X_n,Y_1,\ldots,Y_n,W\) to be mutually independent, with each \(X_i\) distributed as \(X\) and each \(Y_i\) distributed as \(Y\). Then
\[
\operatorname{Var}\left(\sum_{i=1}^nX_i\right)<\operatorname{Var}\left(\sum_{i=1}^nY_i+W\right).
\]
Applying \eqref{eq:bounded-variance-order} and then convolution additivity of \(R\) produces
\(nR(X)\leq nR(Y)+R(W)\), which is impossible for all sufficiently large \(n\).\footnote{To see this, rearrange to obtain \(n\left(R(X)-R(Y)\right)\leq R(W)\). The left-hand side tends to \(+\infty\) because \(R(X)-R(Y)>0\), whereas \(R(W)\) is fixed, a contradiction.} Thus, \(R(X)\leq R(Y)\). Interchanging \(X\) and \(Y\) yields \(R(X)=R(Y)\).

For every \(v\in[0,+\infty)\), the random variable that takes the values \(\sqrt v\) and \(-\sqrt v\) with probability \(1/2\) each is bounded, centered, and has variance \(v\). Because bounded centered risks with the same variance have the same \(R\)-value, define \(f\colon[0,+\infty)\to\mathbf R\) by \(f(v)=R(X)\) whenever \(X\) is bounded, centered, and has variance \(v\).

If \(X\) and \(Y\) are independent bounded centered random variables with variances \(v\) and \(w\), then \(X+Y\) is bounded and centered and has variance \(v+w\). And so, as \(R\) is additive under convolution,
\[
f(v+w)=R(X+Y)=R(X)+R(Y)=f(v)+f(w).
\]
\eqref{eq:bounded-variance-order} implies that \(f\) is nondecreasing. As \(f\) is additive, \(f(0)=f(0)+f(0)\), so \(f(0)=0\). Set \(c\coloneqq f(1)\). Since \(f\) is nondecreasing, \(c\geq0\). Additivity of \(f\) implies that \(f(q)=cq\) for every
\(q\in\mathbf Q\cap[0,+\infty)\).

For an arbitrary \(v\geq0\), choose nonnegative rational sequences \((q_n)\) and \((r_n)\) such that \(q_n\uparrow v\) and \(r_n\downarrow v\). Monotonicity of \(f\) implies  
\[cq_n=f(q_n)\leq f(v)\leq f(r_n)=cr_n,\]
and taking limits yields \(f(v)=cv\).\footnote{Both bounding sequences converge to \(cv\), because \(q_n\to v\) and \(r_n\to v\). The conclusion then follows from the squeeze theorem.}
\end{proof}

Maintain \Cref{ass:additive-convex-order}, and let \(c\) be the constant supplied by \Cref{lem:bounded-calibration}. Whenever \(p\geq2\), define \(D\colon\mathcal P_0^p\to\mathbf R\) by \(D(\mu)\coloneqq R(\mu)-c\operatorname{Var}(\mu)\). Since variance is additive under convolution on \(\mathcal P^2\), \(D\) is additive under convolution. Moreover, \(D(\mu)\leq R(\mu)\) because \(c\geq0\) and \(\operatorname{Var}(\mu)\geq0\).

\begin{lemma}\label{lem:residual-nonnegative}
Maintain \Cref{ass:additive-convex-order} and suppose that \(p\geq2\). Then \(D(X)\geq0\) for every centered \(X\in L^p\).
\end{lemma}

\begin{proof}[Proof of \Cref{lem:residual-nonnegative}]
Fix a centered \(X\in L^p\). If \(X\) is bounded, \Cref{lem:bounded-calibration} implies \(D(X)=0\). Suppose that \(X\) is not bounded. For each \(M>0\), let \(A_M\coloneqq\{|X|\leq M\}\), the event that \(X\) takes a value in \([-M,M]\), let \(b_M\coloneqq\operatorname{E}[X\mid A_M^c]\), and define \(B_M\coloneqq X\mathbf 1_{A_M}+b_M\mathbf 1_{A_M^c}\) and \(\mathcal G_M\coloneqq\sigma(A_M,X\mathbf 1_{A_M})\). The sigma-field \(\mathcal G_M\) reveals \(X\) on \(A_M\) and only whether \(A_M^c\) occurs on its complement. Hence, \(\operatorname{E}[X\mid\mathcal G_M]=X\) on \(A_M\) and \(\operatorname{E}[X\mid\mathcal G_M]=b_M\) on \(A_M^c\), so \(B_M=\operatorname{E}[X\mid\mathcal G_M]\). For every convex function \(\varphi\colon\mathbf R\to\mathbf R\) with at most linear growth, conditional Jensen's inequality produces
\[
\varphi(B_M)
=
\varphi\left(\operatorname{E}[X\mid\mathcal G_M]\right)
\leq
\operatorname{E}[\varphi(X)\mid\mathcal G_M].
\]
Taking expectations yields \(\operatorname{E}[\varphi(B_M)]\leq\operatorname{E}[\varphi(X)]\), and, hence, \(B_M\preceq_{\mathrm{cx}}X\).

The variable \(B_M\) is bounded and centered. Moreover,
\[\operatorname{E}\left[\left(X-B_M\right)^2\right] = \operatorname{E}\left[\left(X-b_M\right)^2\mathbf 1_{A_M^c}\right] =
\operatorname{E}\left[X^2\mathbf 1_{A_M^c}\right]
-b_M^2\operatorname{P}(A_M^c) 
\leq \operatorname{E}\left[X^2\mathbf 1_{A_M^c}\right] \rightarrow0,
\]
where the equality uses \(b_M\operatorname{P}(A_M^c)=\operatorname{E}[X\mathbf 1_{A_M^c}]\), and the convergence follows from the dominated convergence theorem because \(X\in L^2\). Thus \(B_M\to X\) in \(L^2\), which implies \(\operatorname{Var}(B_M)\to\operatorname{Var}(X)\). 

Because \(R\) is monotone in convex order and \(B_M\preceq_{\mathrm{cx}}X\), \Cref{lem:bounded-calibration} implies
\[
R(X)\geq R(B_M)=c\operatorname{Var}(B_M).
\]
Letting \(M\to+\infty\) and using
\(\operatorname{Var}(B_M)\to\operatorname{Var}(X)\) yields
\(R(X)\geq c\operatorname{Var}(X)\), and so \(D(X)\geq0\).
\end{proof}

\begin{lemma}\label{lem:tail-compression}
Maintain \Cref{ass:additive-convex-order}, and let \(c\) be supplied by \Cref{lem:bounded-calibration}. If \(c=0\), then for every centered \(X\in L^p\) and every \(\varepsilon>0\), there is a centered \(Y\in L^p\) such that \(\lVert Y\rVert_p<\varepsilon\) and \(R(Y)\geq R(X)\). If \(p\geq2\), then for every centered \(X\in L^p\) and every \(\varepsilon,\eta>0\), there is a centered \(Y\in L^p\) such that \(\lVert Y\rVert_p<\varepsilon\) and \(D(Y)\geq D(X)-\eta\).
\end{lemma}

\begin{proof}[Proof of \Cref{lem:tail-compression}]
If \(X\) is bounded, \Cref{lem:bounded-calibration} implies \(R(X)=c\operatorname{Var}(X)\). When \(c=0\), choose \(Y=0\). When \(p\geq2\), choose \(Y=0\) as well, since \(D(X)=0\). Suppose henceforth that \(X\) is not bounded.

Since \(X\) is finite almost surely,
\(\operatorname{P}(|X|\leq M)\to1\) as \(M\to+\infty\). Choose \(M_0>0\) such that
\(\operatorname{P}(|X|\leq M)>0\) for every \(M\geq M_0\). Because \(X\) is not bounded, \(\operatorname{P}(|X|>M)>0\) for every \(M>0\). For every \(M\geq M_0\), define
\(A_M\coloneqq\{|X|\leq M\}\), \(\alpha_M\coloneqq\operatorname{P}(A_M)\), and \(\beta_M\coloneqq\operatorname{P}(A_M^c)=1-\alpha_M\). Define also \(m_M\coloneqq\operatorname{E}[X\mid A_M]\) and \(v_M\coloneqq\operatorname{Var}(X\mid A_M)\).

Define \(Y_M\coloneqq m_M\mathbf 1_{A_M}+X\mathbf 1_{A_M^c}\), which replaces the values of \(X\) on \(A_M\) by their conditional mean \(m_M\) and leaves \(X\) unchanged on \(A_M^c\). As \(\alpha_Mm_M=\operatorname{E}[X\mathbf 1_{A_M}]\), we have \(\operatorname{E}[Y_M]
=
\alpha_Mm_M+\operatorname{E}[X\mathbf 1_{A_M^c}]
=
\operatorname{E}[X]
=
0\), and so \(Y_M\) is centered and belongs to \(L^p\).

Let \(W_M\) be independent of \(X\), with law equal to the conditional law of \(X-m_M\) given \(A_M\). Thus, \(\operatorname{E}[W_M]
=
\operatorname{E}[X-m_M\mid A_M]
=
0\), and \(\operatorname{Var}(W_M)=v_M\). The variable \(W_M\) is bounded because \(|X|\leq M\) on \(A_M\). Moreover, \(Y_M\) is a function of \(X\), so the independence of \(W_M\) and \(X\) implies the independence of \(W_M\) and \(Y_M\).

We next compare \(X\) with \(Y_M+W_M\). Let \(\widetilde X_M\) be independent of \(W_M\) and satisfy \(\operatorname{P}(\widetilde X_M\in B) = \operatorname{P}(X\in B\mid A_M^c)\) for every Borel \(B\subseteq\mathbf R\). For every convex function \(\varphi\colon\mathbf R\to\mathbf R\) with at most linear growth, conditional Jensen's inequality implies
\[
\begin{aligned}
\operatorname{E}\left[\varphi(Y_M+W_M)\right]
&=\alpha_M\operatorname{E}\left[\varphi(X)\mid A_M\right]
+\beta_M\operatorname{E}\left[\varphi(\widetilde X_M+W_M)\right]\\
&\geq\alpha_M\operatorname{E}\left[\varphi(X)\mid A_M\right]
+\beta_M\operatorname{E}\left[\varphi(\widetilde X_M)\right] =\operatorname{E}\left[\varphi(X)\right].
\end{aligned}
\]
Moreover, as both variables are centered, \(X\preceq_{\mathrm{cx}}Y_M+W_M\); whence we deduce \(R(X)\leq R(Y_M+W_M)\), because \(R\) is monotone in convex order. Because \(Y_M\) and \(W_M\) are independent, convolution additivity of \(R\) and \Cref{lem:bounded-calibration} then yield
\[
R(X)
\leq
R(Y_M+W_M)
=
R(Y_M)+R(W_M)
=
R(Y_M)+cv_M.
\tag{2}\label{eq:compression-risk-bound}
\]

\begin{claim}
    The variables \(Y_M\) converge to zero in \(L^p\).
\end{claim}
\begin{proof}
We have \(\lVert Y_M\rVert_p^p
=
\alpha_M|m_M|^p
+
\operatorname{E}\left[|X|^p\mathbf 1_{A_M^c}\right]\). Because \(X\in L^p\) and \(\mathbf 1_{A_M^c}\to0\) almost surely, the dominated convergence theorem implies
\(\operatorname{E}\left[|X|^p\mathbf 1_{A_M^c}\right] \rightarrow 0\).
Similarly, \(X\in L^1\), so
\(\operatorname{E}\left[X\mathbf 1_{A_M^c}\right]\rightarrow0\). Since \(X\) is centered,
\(\alpha_Mm_M
=
\operatorname{E}\left[X\mathbf 1_{A_M}\right]
=
-\operatorname{E}\left[X\mathbf 1_{A_M^c}\right]\). Because \(\alpha_M\to1\), \(m_M\to0\), and so
\(\alpha_M|m_M|^p\to0\). Consequently,
\(Y_M \rightarrow 0\) in \(L^p\).
\end{proof}

If \(c=0\), \eqref{eq:compression-risk-bound} implies \(R(Y_M)\geq R(X)\), while the fact \(Y_M \rightarrow 0\) in \(L^p\) permits \(M\) to be chosen so that \(\lVert Y_M\rVert_p<\varepsilon\).

Suppose now that \(p\geq2\). Then \(X\in L^2\). Because \(X\) and \(Y_M\) are centered and agree on \(A_M^c\),
\[\operatorname{Var}(X)-\operatorname{Var}(Y_M) =
\operatorname{E}[X^2\mathbf 1_{A_M}]
-\alpha_Mm_M^2 = \alpha_M\left(
\operatorname{E}[X^2\mid A_M]-m_M^2
\right) = \alpha_Mv_M.
\tag{3}\label{eq:compression-variance-difference}
\]
From \eqref{eq:compression-risk-bound}, \(R(Y_M)-R(X)\geq-cv_M\). Together with
\eqref{eq:compression-variance-difference}, this implies
\[D(Y_M)-D(X) =
R(Y_M)-R(X)
+c\left(\operatorname{Var}(X)-\operatorname{Var}(Y_M)\right) \geq -cv_M+c\alpha_Mv_M = -c\beta_Mv_M.
\]
Equivalently,
\[
D(Y_M)\geq D(X)-c\beta_Mv_M.
\tag{4}\label{eq:compression-residual-bound}
\]
Moreover,
\[
v_M
=
\operatorname{E}[X^2\mid A_M]-m_M^2
\leq
\operatorname{E}[X^2\mid A_M]
=
\frac{\operatorname{E}[X^2\mathbf 1_{A_M}]}{\alpha_M}
\leq
\frac{\operatorname{E}[X^2]}{\alpha_M},
\]
and so
\[
0\leq\beta_Mv_M
\leq
\frac{\beta_M}{\alpha_M}\operatorname{E}[X^2]
\longrightarrow0,
\]
because \(\alpha_M\to1\) and \(\beta_M\to0\).

\eqref{eq:compression-residual-bound} and the fact that \(Y_M \rightarrow 0\) in \(L^p\) allow us to choose \(M\) large enough so that \(\lVert Y_M\rVert_p<\varepsilon\) and \(D(Y_M)\geq D(X)-\eta\). Set \(Y\coloneqq Y_M\).
\end{proof}

The next lemma shows that a nonnegative statistic bounded above by \(R\) must vanish if it is additive under convolution and its value can be preserved, up to an arbitrarily small loss, on centered random variables with arbitrarily small \(L^p\)-norm.

\begin{lemma}\label{lem:stacking}
Maintain \Cref{ass:additive-convex-order}. Suppose that \(Q\colon\mathcal P_0^p\to[0,+\infty)\) is additive under convolution and satisfies \(Q(\mu)\leq R(\mu)\) for every \(\mu\in\mathcal P_0^p\). Suppose also that, for every centered \(X\in L^p\) and every \(\varepsilon,\eta>0\), there is a centered \(Y\in L^p\) such that \(\lVert Y\rVert_p<\varepsilon\) and \(Q(Y)\geq Q(X)-\eta\). Then \(Q(\mu)=0\) for every \(\mu\in\mathcal P_0^p\).
\end{lemma}

\begin{proof}[Proof of \Cref{lem:stacking}]
Suppose, toward a contradiction, that \(Q(X)=d>0\) for some centered \(X\in L^p\). Applying the hypothesis with \(\varepsilon=2^{-n}\) and \(\eta=d/2\), for each \(n\geq1\) choose a centered \(Y_n\in L^p\) such that
\[
\lVert Y_n\rVert_p<2^{-n}
\qquad\text{and}\qquad
Q(Y_n)\geq\frac d2.
\]
Take the variables \((Y_n)_{n\geq1}\) to be mutually independent, and let \(S_N\coloneqq\sum_{n=1}^NY_n\). Minkowski's inequality implies that \((S_N)\) is Cauchy in \(L^p\).\footnote{For \(m>N\), Minkowski's inequality produces \(\lVert S_m-S_N\rVert_p\leq\sum_{n=N+1}^m\lVert Y_n\rVert_p<\sum_{n=N+1}^{\infty}2^{-n}\), and the final tail converges to zero as \(N\to+\infty\).} Let \(Y\in L^p\) denote its limit. Since \(\lVert U\rVert_1\leq\lVert U\rVert_p\) on a probability space, \(S_N\to Y\) in \(L^1\). Each \(S_N\) is centered, so
\(\left|\operatorname{E}[Y]\right|
=
\left|\operatorname{E}[Y-S_N]\right|
\leq
\lVert Y-S_N\rVert_1
\rightarrow0\), and so \(Y\) is centered.

Fix \(N\), and let \(\mathcal F_N\coloneqq\sigma(Y_1,\ldots,Y_N)\). For every \(m>N\), the variables \(Y_{N+1},\ldots,Y_m\) are independent of \(\mathcal F_N\) and have mean zero. Therefore,
\[\operatorname{E}\left[S_m-S_N\mid\mathcal F_N\right] =  \sum_{n=N+1}^m\operatorname{E}\left[Y_n\mid\mathcal F_N\right] = \sum_{n=N+1}^m\operatorname{E}[Y_n] = 0.\]

Because \(S_m-S_N\to Y-S_N\) in \(L^1\), the \(L^1\)-contraction property of conditional expectation\footnote{For every integrable \(U\), \(\lVert\operatorname{E}[U\mid\mathcal F_N]\rVert_1\leq\lVert U\rVert_1\). Consequently, if \(U_m\to U\) in \(L^1\), then \(\operatorname{E}[U_m\mid\mathcal F_N]\to\operatorname{E}[U\mid\mathcal F_N]\) in \(L^1\).} yields
\(\operatorname{E}\left[Y-S_N\mid\mathcal F_N\right]=0\). Thus, \(S_N=\operatorname{E}[Y\mid\mathcal F_N]\), so conditional Jensen's inequality delivers \(S_N\preceq_{\mathrm{cx}}Y\). 

Because \(R\) is monotone in convex order, \(S_N\preceq_{\mathrm{cx}}Y\) implies \(R(Y)\geq R(S_N)\). The bound \(Q\leq R\) implies \(R(S_N)\geq Q(S_N)\). Because \(Y_1,\ldots,Y_N\) are independent, convolution additivity of \(Q\) yields \(Q(S_N)=\sum_{n=1}^NQ(Y_n)\). Consequently,
\[
R(Y)\geq R(S_N)\geq Q(S_N)
=\sum_{n=1}^NQ(Y_n)
\geq\frac{Nd}{2}.
\]
This inequality cannot hold for every \(N\) because \(R(Y)\) is a finite real number. Accordingly, \(Q(X)\) cannot be positive for any centered \(X\in L^p\). Since \(Q\) is nonnegative, \(Q(\mu)=0\) for every \(\mu\in\mathcal P_0^p\).
\end{proof}

\begin{lemma}\label{lem:subquadratic-calibration}
Maintain \Cref{ass:additive-convex-order} and suppose that \(1\leq p<2\). Then the constant \(c\) supplied by \Cref{lem:bounded-calibration} equals zero.
\end{lemma}

\begin{proof}[Proof of \Cref{lem:subquadratic-calibration}]
Choose \(r\in(p,2)\), and let \(\gamma\coloneqq\left(\sum_{n=1}^{\infty}2^{-rn}\right)^{-1}\). Define a symmetric random variable \(H\) by
\[
\Pr(H=2^n)=\Pr(H=-2^n)=\frac{\gamma}{2}2^{-rn},
\]
for every \(n\geq1\). Then
\[
\operatorname{E}[|H|^p]
=\gamma\sum_{n=1}^{\infty}2^{(p-r)n}
<+\infty
\qquad \text{and} \qquad
\operatorname{E}[H^2]
=\gamma\sum_{n=1}^{\infty}2^{(2-r)n}
=+\infty.
\]
Thus, \(H\in L^p\) and \(H\) is centered, but \(H\notin L^2\).

For \(N\geq1\), define \(H_N \coloneqq H\mathbf 1_{\{|H|\leq2^N\}}\) and \(\mathcal H_N \coloneqq\sigma\left(\{|H|\leq2^N\},H\mathbf 1_{\{|H|\leq2^N\}}\right)\). Symmetry of \(H\) implies \(\operatorname{E}[H\mid |H|>2^N]=0\), so \(H_N=\operatorname{E}[H\mid\mathcal H_N]\). Consequently, \(H_N\preceq_{\mathrm{cx}}H\). The variable \(H_N\) is bounded and centered, and the monotone convergence theorem implies \(\operatorname{Var}(H_N)\to+\infty\) (via \(H_N^2
=
H^2\mathbf 1_{\{|H|\leq2^N\}}
\uparrow H^2\) plus \(\operatorname{Var}(H_N)
=
\operatorname{E}[H_N^2]\)). 

Thus, \(R(H)\geq R(H_N)=c\operatorname{Var}(H_N)\), and  since \(R(H)\) is finite, \(c=0\).\footnote{If \(c>0\), then \(c\operatorname{Var}(H_N)\to+\infty\), contradicting the fact that \(R(H)\) is a finite real number. Since \(c\geq0\), we conclude that \(c=0\).}
\end{proof}

\begin{proof}[Proof of \Cref{thm:second-moment-dichotomy}]
Suppose first that \(1\leq p<2\). By \Cref{lem:subquadratic-calibration}, \(c=0\). By \Cref{lem:bounded-calibration}, \(R(0)=0\). For every centered \(X\in L^p\), \(0\preceq_{\mathrm{cx}}X\), so \(R(X)\geq0\). Set \(Q\coloneqq R\). Then \(Q\) is nonnegative and additive under convolution, and \(Q\leq R\) trivially. The \(c=0\) conclusion of \Cref{lem:tail-compression} produces the small-\(L^p\)-norm approximation required by \Cref{lem:stacking}. Therefore, \(R=0\).

Suppose now that \(2\leq p<+\infty\). The statistic \(D\) is additive under convolution and satisfies \(D\leq R\), while by \Cref{lem:residual-nonnegative}, \(D\geq0\). The \(p\geq2\) conclusion of \Cref{lem:tail-compression} produces the small-\(L^p\)-norm approximation required by \Cref{lem:stacking}. Apply that lemma with \(Q\coloneqq D\) to get \(D=0\). Therefore,
\(R(\mu)=c\operatorname{Var}(\mu)\) for every
\(\mu\in\mathcal P_0^p\).

Conversely, the zero map is additive under convolution and monotone in convex order. When \(p\geq2\), variance has both properties, so every map
\(\mu\mapsto c\operatorname{Var}(\mu)\) with \(c\geq0\) satisfies
\Cref{ass:additive-convex-order}.
\end{proof}

\section{The full domain and certainty equivalents}

\begin{proof}[Proof of \Cref{cor:full-classification}]
Define \(a\colon\mathbf R\to\mathbf R\) by \(a(x)\coloneqq\Phi(\delta_x)\). Since \(\delta_x*\delta_y=\delta_{x+y}\), additivity implies \(a(x+y)=a(x)+a(y)\) for every \(x,y\in\mathbf R\). The restriction of \(\Phi\) to \(\mathcal P_0^p\) is real-valued, additive, and monotone. Thus, \Cref{thm:second-moment-dichotomy} applies to that restriction.

Fix \(\mu\in\mathcal P^p\), let \(m_\mu\coloneqq m(\mu)\), and define \(\overline\mu\coloneqq\mu*\delta_{-m_\mu}\). Then \(\overline\mu\in\mathcal P_0^p\), \(\mu=\overline\mu*\delta_{m_\mu}\), and \(\operatorname{Var}(\overline\mu)=\operatorname{Var}(\mu)\) whenever \(p\geq2\). Hence, \(\Phi(\mu)=\Phi(\overline\mu)+a(m_\mu)\). If \(p<2\), \Cref{thm:second-moment-dichotomy} implies \(\Phi(\overline\mu)=0\); and if \(p\geq2\), it implies \(\Phi(\overline\mu)=c\operatorname{Var}(\mu)\) for some \(c\geq0\). Combining these two cases delivers the stated formulas.

Conversely, means add under convolution, and variances add across independent risks. Convex order preserves the mean and weakly increases variance, so every map in \Cref{cor:full-classification} has the stated properties.
\end{proof}

\begin{proof}[Proof of \Cref{cor:certainty-equivalent}]
Define the risk premium \(\rho\colon\mathcal P^p\to\mathbf R\) by \(\rho(\mu)\coloneqq m(\mu)-C(\mu)\). The map \(\rho\) is additive under convolution. If \(\mu\preceq_{\mathrm{cx}}\nu\), then \(m(\mu)=m(\nu)\) and \(C(\mu)\geq C(\nu)\), so \(\rho(\mu)\leq\rho(\nu)\). Moreover, \(\rho(\delta_x)=0\) for every \(x\in\mathbf R\). Applying \Cref{cor:full-classification} to \(\rho\), therefore, eliminates the additive mean term. If \(p<2\), then \(\rho(\mu)=0\) and \(C(\mu)=m(\mu)\). If \(p\geq2\), then \(\rho(\mu)=c\operatorname{Var}(\mu)\) for some \(c\geq0\), and hence \(C(\mu)=m(\mu)-c\operatorname{Var}(\mu)\).

Conversely, expected value satisfies the hypotheses when \(1\leq p<2\). When \(2\leq p<+\infty\), every mean-variance certainty equivalent with \(c\geq0\) satisfies certainty, independent additivity, and aversion to mean-preserving spreads.
\end{proof}

\bibliographystyle{plainnat}
\bibliography{references}

\end{document}